\documentclass[12pt]{article}

\usepackage{dsfont}

\usepackage{calrsfs}

\usepackage{amsmath, amsthm, amssymb, amsfonts, enumerate}

\usepackage[longnamesfirst, comma]{natbib}
\usepackage{float}

\usepackage[usenames]{color}

\usepackage{tikz}

\usetikzlibrary{arrows,decorations.pathmorphing,backgrounds,positioning,fit,automata}

\usetikzlibrary{arrows}
\usetikzlibrary{petri}
\usetikzlibrary{topaths}

\usepackage{indentfirst,calc,euscript}
\usepackage{setspace}
\usepackage[reals]{layout}
\usepackage{xr}
\usepackage{amscd}

\usepackage{sgame}
\usepackage{subfigure}

\newcommand{\B}{\operatorname{B}}

\newcommand{\BETA}{\operatorname{Beta}}

\newcommand{\BINOMIAL}{\operatorname{Binom}}

\newcommand{\Exp}{\operatorname{Exp}}

\newcommand{\expo}{\operatorname{e}}

\newcommand{\GAMMA}{\operatorname{Gamma}}

\newcommand{\Geom}{\operatorname{Geom}}

\newcommand{\Med}{\operatorname{Med}}

\newcommand{\Poisson}{\operatorname{Poisson}}

\newcommand{\diff}{\ \mathrm{d}}

\newcommand{\ignore}[1]{}

\newtheorem{theorem}{Theorem}

\newtheorem{claim}[theorem]{Claim}

\newtheorem{corollary}[theorem]{Corollary}

\newtheorem{lemma}[theorem]{Lemma}

\newtheorem{proposition}[theorem]{Proposition}

\theoremstyle{definition}

\newtheorem{case}{Case}

\newtheorem{remark}[theorem]{Remark}

\numberwithin{equation}{section}

\numberwithin{theorem}{section}

\begin{document}
\title{A Colonel Blotto Gladiator Game}
\author{Yosef Rinott%
\thanks{Partially supported by the
Israel Science Foundation grant No. 473/04 }\\
   Department of Statistics\\
   and Center for the Study of Rationality\\
   Hebrew University of Jerusalem\\
   Mount Scopus\\
   Jerusalem 91905, Israel\\
   \texttt{rinott@mscc.huji.ac.il}
   \and
   Marco Scarsini\\
   Dipartimento di Economia e Finanza\\
   LUISS\\
   Viale Romania 12\\
   I--00197 Roma, Italy\\
   \texttt{marco.scarsini@luiss.it}
   \and
   Yaming Yu\\
   Department of Statistics\\
   University of California, Irvine\\
   CA 92697-1250, USA\\
   \texttt{yamingy@uci.edu}
   }

\date{\today}

\maketitle

\thispagestyle{empty}

\pagebreak

\begin{abstract}

We consider a stochastic version of the well-known Blotto game, called the \emph{gladiator game}. In this zero-sum allocation game two teams of gladiators engage in a sequence of one-to-one fights in which the probability of winning is a function of the gladiators' strengths. Each team's strategy consist the allocation of its total strength among its gladiators.  We find the Nash equilibria and the value of this class of games and show how they depend on the total strength of teams and the number of gladiators in each team. To do this, we study interesting majorization-type probability inequalities concerning linear combinations of Gamma random variables. Similar inequalities have been used in models of telecommunications and research and development.



\bigskip
\noindent \emph{OR/MS subject classification}: Primary: games/group decisions--noncooperative; secondary: probability--distribution comparisons.

\end{abstract}

\section{Introduction}\label{se:intro}

\citet{Bor:CRAS1921} proposed a game, later dubbed \emph{Colonel Blotto game} by \citet{GroWag:RAND1950}.
In this game Colonel Blotto and his enemy each have a given (possibly unequal) amount of resources, that have to be allocated to $n$ battlefields. The side that allocates more resources to   field $j$ is the winner in this field and  gains a positive amount $a_{j}$ which the other side loses. The war is won by the army that obtains the largest total gain.

The relevance of Borel precursory insight in the theory of games was discussed in an issue of \emph{Econometrica} that contains three papers by Borel, including the translation of the 1921 paper  \citep{Bor:E1953}, two notes by \citet{Fre:E1953a,Fre:E1953b} and one by \citet{von:E1953}.

\citet{BorVil:GV1938} proposed a solution to the game when the two enemies have an equal amount of resources and there are $n=3$ battlefields. The problem was then taken up by several authors, including several other famous mathematicians.
\citet{GroWag:RAND1950, Gro:RAND1950} provided the solution for a generic $n$, keeping the amount of resources equal and the gain in each battlefield constant ($a_{i}=a_{j}$).
\citet{Bla:NRLQ1954,Bla:NRLQ1958} considered the case where the payoff to Colonel Blotto in each battlefields is an increasing function of his resources and a decreasing function of his enemy's resources. \citet{Bel:SIAMR1969} showed the use of dynamic programming to solve the Blotto game.
\citet{ShuWeb:NRLQ1981} studied a more complex model where there exist complementaries among the fields being defended. In this case the total payoff depends on the relative value of capturing various configurations of targets.
\citet{Rob:ET2006} used $n$-copulas to determine the mixed equilibrium of the game under general conditions on the amount of resources for each player. His analysis is based on an interesting analogy with the theory of all-pay auctions (see also  \citet{Wei:mimeo2005} for the equilibrium of the game and \citet{SahPer:ET2006} for the connection between all-pay auctions and allocation of electoral promises).

\citet{Har:IJGT2008} considered a discrete version of the Blotto game, where player A has $A$ alabaster marbles and player B has $B$ black marbles. The players are to distribute their marbles into $K$ urns. One urn is chosen at random and the player with the largest number of marbles in the chosen urn wins the game.
In another version of the game, called \emph{Colonel Lotto game}, each player has $K$ urns where she can distribute her marbles. Two urns (one for each player) are chosen at random and the urn with the larger number of marbles determines the winner. The discrete Colonel Blotto game and the Colonel Lotto game have the same value. In a third version, called \emph{General Lotto game}, given $a,b>0$, player A chooses  a nonnegative integer-valued random variable $X$ with expectation $\mathbb{E}[X] = a$ and player B chooses  a nonnegative integer-valued random variable $Y$ with expectation $\mathbb{E}[Y] = b$. The payoff for A is $\mathbb{P}(X >Y)-\mathbb{P}(X <Y)$, where $X$ and $Y$ are assumed independent. The value of the game and the optimal strategies are determined.

Other authors who dealt with the Blotto game and its applications include, for instance,  \citet{Tuk:E1949,  SioWol:CTG3PUP1957, Fri:OR1958, CooRes:SIAMR1967, Pen:OR1971, Heu:TCS2001,  Kva:JET2007, AdaMat:EL2009,   Pow:GEB2009, GolPag:PC2009} and many more.
We refer to  \citet{KovRob:CESIFO2010, ChoKovShe:CESIFO2010}  for some history of the Colonel Blotto game and a good list of references.

In this paper we deal with a stochastic version of the Colonel Blotto game, called \emph{gladiator game} by \citet{KamLukNel:AJS1984}.
In their model two teams of gladiators engage in a sequence of one-to-one fights.  Each gladiator has a strength parameter.  When two gladiators fight, the ratio of their strengths determines the odds of winning.  The loser dies and the winner retains his strength and is ready for a new duel. The team that is wiped out loses. Each team chooses once and for all at the beginning of the game the order in which gladiators go to the arena.

We construct a zero-sum two-team game where each team also has to allocate a fixed total strength among its players.  The payoff is linear in the probability of winning.  We find the Nash equilibria and compute the value of the game.  The main results are:
\begin{enumerate}[(i)]
\item the order according to which gladiators fight has no relevance, moreover knowing the order chosen by the opponent team does not provide any advantage;

\item the stronger team always splits its strength uniformly among its gladiators, whereas the weaker team splits the strength uniformly among a subset of its gladiators;

\item
when the two teams have roughly equal total strengths, the optimal strategy for the weaker team is to divide its total strength equally among all its members;

\item
when the total strength of one team is much larger than that of the other,  the weaker team should concentrate all the strength on a single member.
\end{enumerate}

\citet{DeSDeMDeB:DAM2006} consider a dice game that has some analogies with ours. Both players can choose one of many dice having $n$ faces and such that the total number of pips on the faces of the die is $\sigma$.  The two dice are tossed and the player with the highest score wins a dollar.

The model described below for the probability that gladiator $i$ defeats $j$, is equivalent, with different parametrization,  to the well-known Rasch model in educational statistics, \citep{Ras:UChicagoP1960}, in which the probability of correct response of subject $i$ to item $j$ is $\expo^{\alpha_i-\beta_j}/(1+\expo^{\alpha_i-\beta_j})$ \citep[see][for a recent mathematical study of Rasch models]{Lau:RMA2008}.
A similar model has been used also in the theory of contests proposed by \citet{Tul:TAMUP1980}, as will be described in Section~\ref{se:extensions}.

Finding the Nash equilibria of the gladiator game involves an analysis of the probability of winning. The key step is a result in \citet{KamLukNel:AJS1984} that translates the calculation of this probability into an inequality involving the sum of independent but not necessarily identically distributed exponential random variables.

The main theorems are demonstrated through interesting and hard probability inequalities, whose proofs are of independent interest and turned out to be more complicated than expected.  Much of the paper consists of these proofs. We rely on \citet{SzeBak:PTRF2003} for some of the technical machinery.  The problem is cast as a minimization problem involving convolutions of exponential variables and is solved by perturbation arguments.  A key identity, derived using Laplace transforms, directs our perturbation arguments to the analysis of the modal location of Gamma convolutions.

Our inequalities are related to majorization type inequalities for probabilities of the form $\mathbb{P}(Q <t)$, where $Q$ is a linear combinations of Exponential or Gamma variables, that appear in \citet{BocDiaHufPer:CJS1987,DiaPer:IMSLN1990,SzeBak:PTRF2003} and in
\citet{Tel:ETT1999, JorBoc:ITW2003, AbbHuaTel:ArXiv2011}.
The motivation in the last three papers, and numerous others, is the performance of some wireless systems that depends on the coefficients of the linear combination  $Q$.
For stochastic comparisons between such linear combinations see \citet{Yu:AAP2008, Yu:B2011}  and references therein.

Linear combinations of exponential variables appear in various other applications. For instance \citet{LipMcC:RANDJE1987} consider a two-firm model in which learning is stochastic and the research race is divided into a finite number $N$ of stages, each having an exponential completion
date. The invention is discovered at the completion of the $N$-th stage. If the exponential times for one firm have parameters that can be controlled by the firms subject to constraints, then our results apply to the problem of best response and equilibrium allocation strategies for such races.

Finally, it is well known that the first passage time from $0$ to $N$ of a birth and death process on the positive integers is distributed as a linear combination of exponential random variables, with coefficients determined by the eigenvalues of the process' generator. For a clear statement, a probabilistic proof, and further references see
\citet{DiaMic:JTP2009}.  This allows one to consider R\&D type races in which one can also move backwards, and applies, for example, to the study of queues, where one compares the time until different systems reach a given queue size.

The paper is organized as follows.
In Section~\ref{se:model} we describe the model.  In Section~\ref{se:main} we determine the  Nash equilibria and the value of the game for different values of the parameters.  Section~\ref{se:inequalities} contains the main probability inequalities used to compute the equilibria. Section~ \ref{se:proofs} is devoted to the proofs of the main results. Section~\ref{se:monotonicity} deals with some monotonicity properties, that follow from our main result and have some interest \emph{per se}. Finally Section~\ref{se:extensions} contains some extensions and open problems.

\section{The model}\label{se:model}

We formalize the model described in the Introduction.
Two teams of gladiators fight each other according to the following rules. Team $A$ is an ordered set $\{A_{1}, \dots, A_{m}\}$ of $m$ gladiators and team $B$ is an ordered set $\{B_{1}, \dots, B_{n}\}$ of  $n$ gladiators. The numbers $m,n$ and the orders of the gladiators in the two teams are exogenously given.
At any given time, only two gladiators fight, one for each team. At the end of each fight only one gladiator survives. In each team gladiators go to fight according to the exogenously given order. First gladiators $A_{1}$ and $B_{1}$ fight.  The winner remains in the arena and fights the following gladiator of the opposing team. Assume that for  $i<m$ and $j<n$  at some point, $A_{i}$ fights $B_{j}$. If $A_{i}$ wins, the following fight will be between $A_{i}$ and $B_{j+1}$; if $A_{i}$ loses, the following fight will be between $A_{i+1}$ and $B_{j}$. The process goes on until a team is wiped out. The other team is then proclaimed the winner. So if at some point, for some $i \le m$, gladiator $A_{i}$ fights $B_{n}$ and wins, then team $A$ is the winner. Symmetrically if, for some $j \le n$, $A_{m}$ fights $B_{j}$ and loses, then team $B$ is the winner.

Team $A$ has total strength $c_{A}$ and team $B$ has total strength $c_{B}$. The values $c_{A}$ and $c_{B}$ are exogenously given. Before fights start the coach of each team decides how to allocate the total strength to the gladiators of the team. These decisions are simultaneous and cannot be altered during the play. Let  $\boldsymbol{a}=(a_{1}, \dots, a_{m})$ and  $\boldsymbol{b}=(b_{1}, \dots, b_{n})$ be the strength vectors of team $A$ and $B$, respectively. This means that in team $A$ gladiator $A_{i}$ gets strength $a_{i}$ and in team $B$ gladiator $B_{j}$ gets strength $b_{j}$. The vectors $\boldsymbol{a}, \boldsymbol{b}$ are nonnegative and such that
\[
\sum_{i=1}^{m} a_{i} = c_{A}, \quad \sum_{j=1}^{n} b_{j} = c_{B},
\]
namely, each coach distributes all the available strength among the gladiators of his team.

When a gladiator with strength $a$ fights a gladiator with strength $b$, the first defeats the second with probability
\begin{equation}\label{eq:aoveraplusb}
\frac{a}{a+b},
\end{equation}
all fights being independent. When a gladiator wins a fight his strength remains unaltered. The rules of the play and its parameters, i.e., the teams $A$ and $B$ and the strengths  $c_{A}, c_{B}$, are common knowledge.
Call $G_{m,n}(\boldsymbol{a},\boldsymbol{b})$ the probability that team $A$ with strength vector $\boldsymbol{a}$ wins over team $B$ with strength vector $\boldsymbol{b}$.

The above model gives rise to the  zero-sum two-person game
\begin{equation}\label{eq:game}
\mathcal{G}(m,n,c_{A},c_{B}) = \langle \mathcal{A}(m,c_{A}), \mathcal{B}(n,c_{B}), H_{m,n} \rangle
\end{equation}
in which team $A$ chooses $\boldsymbol{a} \in \mathcal{A}(m,c_{A})$ and $B$ chooses $\boldsymbol{b} \in \mathcal{B}(n,c_{B})$, where
\begin{align}
\mathcal{A}(m,c_{A})&=\left\{(a_1,\dots,a_m) \in \mathbb{R}^{m}_{+}:  \sum_{i=1}^m a_i=c_{A}\right\}, \label{eq:mathcalA}\\
\mathcal{B}(n,c_{B})&=\left\{(b_1,\dots,b_n) \in \mathbb{R}^{n}_{+} :  \sum_{i=1}^n b_i=c_{B}\right\}, \label{eq:mathcalB} \\
H_{m,n} &=G_{m,n} - \frac{1}{2}. \label{eq:payoff}
\end{align}

The payoff of team $A$ is then its probability of victory $G_{m,n}(\boldsymbol{a},\boldsymbol{b})$ minus $1/2$. We subtracted $1/2$ to make the game zero-sum.

As will be shown in Remark~\ref{re:marco} below, other models with different rules of engagement for the gladiators give rise to the same zero-sum game.

\section{Main results}\label{se:main}
Consider the game $\mathcal{G}$ defined in \eqref{eq:game}. The action $\boldsymbol{a}^*$ is a best response  against $\boldsymbol{b}$  if
\[
\boldsymbol{a}^{*} \in \arg \max_{\boldsymbol{a} \in {\mathcal{A}}}H_{m,n}(\boldsymbol{a},\boldsymbol{b}).
\]

A pair of actions $(\boldsymbol{a}^*,\boldsymbol{b}^*)$ is a \emph{Nash equilibrium} of the game $\mathcal{G}$ if
\[
H_{m,n}(\boldsymbol{a},\boldsymbol{b}^{*}) \le H_{m,n}(\boldsymbol{a}^{*},\boldsymbol{b}^{*}) \le H_{m,n}(\boldsymbol{a}^{*},\boldsymbol{b}), \quad\text{for all}\ \boldsymbol{a} \in \mathcal{A}(m,c_{A})\ \text{and}\ \boldsymbol{b} \in \mathcal{B}(n,c_{B}).
\]

A pair of actions $(\boldsymbol{a}^*,\boldsymbol{b}^*)$ is a \emph{minmax solution} of the game $\mathcal{G}$ if
\[
\max_{\boldsymbol{a} \in \mathcal{A}(m,c_{A})} \min_{\boldsymbol{b} \in \mathcal{B}(n,c_{B})} H_{m,n}(\boldsymbol{a},\boldsymbol{b}) = \min_{\boldsymbol{b} \in \mathcal{B}(n,c_{B})}  \max_{\boldsymbol{a} \in \mathcal{A}(m,c_{A})} H_{m,n}(\boldsymbol{a},\boldsymbol{b}) = H_{m,n}(\boldsymbol{a}^{*},\boldsymbol{b}^{*}).
\]

Since we are dealing with a finite zero-sum game, Nash equilibria and minmax solutions coincide \citep[see, e.g.,][Proposition~22.2]{OsbRub:MITP1994}.
The quantity $H_{m,n}(\boldsymbol{a}^{*},\boldsymbol{b}^{*})$ is called the \emph{value} of the game $\mathcal{G}$.

The next theorem characterizes the structure of Nash equilibria of the game $\mathcal{G}(m,n,c_{A},c_{B})$.

\begin{theorem}\label{th:Nash}
Consider the game $\mathcal{G}(m,n,c_{A},c_{B})$ defined in \eqref{eq:game}. Assume that $c_{A} \le c_{B}$.
\begin{enumerate}[{\rm (a)}]
\item\label{it:th:Nash-a}
There exists an equilibrium strategy profile $(\boldsymbol{a}^{*}, \boldsymbol{b}^{*})$  of $\mathcal{G}$ such that for some $J \subseteq \{1, \dots, m \}$ we have
\begin{align}\label{eq:equila}
a^{*}_{i} & = c_{A}/|J| \ \text{ for }\ i  \in J, \quad a^{*}_{i} = 0 \ \text{ for }\ i \in  J^{c}, \\
b^{*}_{i} & = c_{B}/n \ \text{ for }\ i \in  \{1, \dots, n\}. \label{eq:equilb}
\end{align}
Moreover, all pure equilibria are of this form.

\item\label{it:th:Nash-b}
If
\begin{equation}\label{eq:inequalitycasimcb}
c_{B} \le \frac{n}{n-1} c_{A},
\end{equation}
then $J = \{1, \dots, m \}$, so that  $a_{1}^{*}=\dots=a_{m}^{*} = c_{A}/m$ and $b_{1}^{*}=\dots=b_{n}^{*}=c_{B}/n$.

\item\label{it:th:Nash-c}
If
\begin{equation}\label{eq:inequalitycasmallercb}
c_{B} \ge \frac{3n}{2(n-1)}c_{A},
\end{equation}
then $J = \{ i \}$, that is  $a_{i}^{*}=c_{A}$ for some $i \in \{1,\dots,m\}$ and $a_{j}^{*}=0$ for all $j \ne i$, and
 $b_{1}^{*}=\dots=b_{n}^{*}=c_{B}/n$.

\item\label{it:th:Nash-d}
Let $t_0=1.256431\cdots$ be the root of the equation $\expo^t = 1+2t$.  Then for fixed $m$, and $c_A$ and $c_B$
such that $c_B > t_0 c_A$, the same conclusion as in \eqref{it:th:Nash-c} holds if $n$ is sufficiently large.

\end{enumerate}

\end{theorem}

Theorem~\ref{th:Nash} shows that if a vector $(\boldsymbol{a}^{*},\boldsymbol{b}^{*})$ is an equilibrium, then so is any permutation of $\boldsymbol{a}^{*}$ or $\boldsymbol{b}^{*}$. Moreover the team with the highest total strength always divides it equally among its members, whereas the other team divides its strength equally among a subset of its members. This subset coincides with the whole team if the total strengths of the two teams are similar, and it reduces to one single gladiator if the team has a much lower strength than the other team (see Figures~\ref{fi:plotmeqn}, \ref{fi:plotfixedn}, and \ref{fi:plotfixedm}).

\begin{center}
FIGURES~\ref{fi:plotmeqn}, \ref{fi:plotfixedn},  AND  \ref{fi:plotfixedm} ABOUT HERE
\end{center}

For $n=1$, i.e., when team $B$ has a single player, equal strength is always team $A$'s best strategy.

In order to compute the value of the game $\mathcal{G}(m,n,c_{A},c_{B})$, we need the regularized incomplete beta function
\begin{equation}\label{eq:incompletebeta}
I(x,\alpha,\beta) = \frac{1}{\B(\alpha, \beta)}\int_{0}^{x} t^{\alpha-1}(1-t)^{\beta-1} \diff t,
\end{equation}
where
\[
\B(\alpha, \beta)=\int_{0}^{1} t^{\alpha-1}(1-t)^{\beta-1} \diff t = \frac{\Gamma(\alpha)\Gamma(\beta)}{\Gamma(\alpha+\beta)}.
\]
When $\alpha$ and $\beta$ are integers, then
\begin{equation}\label{eq:binombeta}
I(x,\alpha,\beta) = \sum_{j=\alpha}^{\alpha+\beta-1} \binom{\alpha+\beta-1}{j}x^{j}(1-x)^{\alpha+\beta-1-j}.
\end{equation}

For properties of incomplete beta functions see, for instance, \citet{OlvLozBoiCla:NISTHMF2010}.

\begin{theorem}\label{th:value}
Consider the game $\mathcal{G}(m,n,c_{A},c_{B})$. Assume that $c_{A} \le c_{B}$.

\begin{enumerate}[\rm (a)]
\item\label{it:th:value-a}
The value of the game is
\begin{equation}\label{eq:valuegeneral}
\frac{1}{2}-I\left(\frac{rc_{B}}{rc_{B}+nc_{A}},r,n\right),
\end{equation}
where $r$ is the number of positive $a_{i}^{*}$ in the vector $\boldsymbol{a}^{*}$. In particular

\item\label{it:th:value-b}
if \eqref{eq:inequalitycasimcb} holds, then the value of the game is
\begin{equation}\label{eq:valueequal}
\frac{1}{2}-I\left(\frac{mc_{B}}{mc_{B}+nc_{A}},m,n\right),
\end{equation}

\item\label{it:th:value-c}
if \eqref{eq:inequalitycasmallercb} holds,  then the value of the game is
\begin{equation}\label{eq:valueunequal}
\frac{1}{2}-I\left(\frac{c_{B}}{c_{B}+nc_{A}},1,n\right).
\end{equation}
\end{enumerate}
\end{theorem}

In general, to compute the value of the game, one only needs to maximize (\ref{eq:valuegeneral}) over $r=1,\ldots, m$;
any maximizing $r$ gives an optimal strategy for team $A$.  Figure~\ref{fi:plotvariousr} shows the value of the game as $c_{B}$ varies. Different values of $c_{B}$ imply different numbers of positive $a_{i}^{*}$.

\begin{center}
FIGURE~\ref{fi:plotvariousr}  ABOUT HERE
\end{center}

\section{Probability inequalities}\label{se:inequalities}

We say that $X\sim \Exp(1)$ if $X$ has a standard exponential distribution, i.e., $\mathbb{P}(X > x) = \expo^{-x}$ for $x>0$.

The main theorems of this paper rely on the following result.

\begin{proposition}\label{pr:kaminsky}[\citet{KamLukNel:AJS1984}]
The probability $G_{m,n}(\boldsymbol{a},\boldsymbol{b})$ of team $A$ defeating $B$ is
\begin{equation}\label{eq:defG}
G_{m,n}(\boldsymbol{a},\boldsymbol{b}) = \mathbb{P}\left(\sum_{i=1}^m a_i X_i > \sum_{j=1}^n b_j Y_j \right),
\end{equation}
where $X_{1}, \dots, X_{m}, Y_{1}, \dots, Y_{n}$ are i.i.d. random variables, with $X_{1}\sim \Exp(1)$.
\end{proposition}
\begin{remark}\label{re:marco}
The implication of Proposition~\ref{pr:kaminsky} is that two vectors $\boldsymbol{a}, \boldsymbol{a}'$ of strengths that are equal up to a permutation produce the same probability of victory, that is, the same payoff function \eqref{eq:payoff}. The same holds for two vectors  $\boldsymbol{b}, \boldsymbol{b}'$.
Therefore various models, with different rules for the order in which gladiators fight, give rise to the same game  \eqref{eq:game}. This happens, for instance, in a model where the winning gladiator does not stay in the arena to fight the following opponent, but, rather, goes to the bench at the end of his team's queue, and comes back to fight when his turn comes. This happens also when, at the end of each fight, each coach chooses one of the living gladiators in his team at random and sends him to fight. Basically, provided the allocations of strength in the two teams is decided simultaneously at the beginning and is not modified throughout, any rule governing the order of descent of gladiators in the arena leads to the same game \eqref{eq:game}. This is true also for nonanticipative  rules that depend on the history of the battles so far. The key assumption for this is the fact that a winning gladiator does not lose (or gain) any strength after a victorious battle. This is parallel to the lack-of-memory property in many reliability models, and explains why the probability of winning \eqref{eq:defG} involves sums of exponential random variables.

Note that the main result (Theorem~\ref{th:Nash}) does not go through if the allocations can also be decided dynamically as battles unfold. In this case the resulting game is more complicated and optimal allocations may change according to the observed history. For instance consider the case where $c_{B}$ is slightly larger than $c_{A}$. At the beginning, suppose team $B$ spreads the strength uniformly across all its players.  If team $B$ keeps losing some battles, then it may become optimal to spread the strength among only a subset of the surviving players.

\end{remark}

The following theorem is the main tool to prove Theorem~\ref{th:Nash}.

\begin{theorem}
\label{th:minimizer}
Let $X_1,\dots, X_m$ and $Y_1,\dots, Y_n$, $m, n\geq 1$, be i.i.d. random variables with $X_{1}\sim \Exp(1)$.  For fixed $b>0$,  let $\mathcal{A}$ be as in \eqref{eq:mathcalA} and let
\[
(a_1^*,\dots, a_m^*)\in \arg \min_{\boldsymbol{a}\in \mathcal{A}(m,m)}\mathbb{P}\left(\sum_{i=1}^m a_i X_i \leq b\sum_{j=1}^n Y_j\right).
\]
Then
\begin{enumerate}[\rm (a)]
\item\label{it:th:minimizer-a}
all nonzero values among $a_1^*,\ldots, a_m^*$ are equal;

\item\label{it:th:minimizer-b}
if $m \ge (n-1)b$, then $a^*_1=\cdots =a^*_m=1$;

\item\label{it:th:minimizer-c}
if $m\leq 2(n-1)b/3$, then $a^*_i=m$ for a single $i$,\, $1\leq i\leq m$, and $a^*_j=0$, for $j\neq i$.
\end{enumerate}
\end{theorem}

\section{Proofs of the main results}\label{se:proofs}

The long path to the proof of Theorem~\ref{th:Nash} goes through the following steps: first we provide a short proof of  Proposition~\ref{pr:kaminsky} for the sake of completeness. Then we state and prove three lemmas needed for the proof of Theorem~\ref{th:minimizer}. Then we prove Theorem~\ref{th:minimizer}, and, resorting to it, we finally prove Theorem~\ref{th:Nash}.

\begin{proof}[Proof of Proposition~\ref{pr:kaminsky}]
First note that if $X$, $Y$ are i.i.d.~random variables with $X\sim \Exp(1)$, then $\mathbb{P}(aX>bY)= a/(a+b)$. Therefore, one can see a duel between gladiators $i$ and $j$ as a competition in which the probability of winning is the probability of living longer, when their lifetimes are $a_iX$ and $b_jY$, respectively. At the end of a duel, the winner's remaining lifetime is as good as new by the memoryless property of exponential random variables, corresponding to the fact that the strength of a winner remains unaltered. The teams' total lives are $\sum_{i=1}^m a_i X_i$ and $\sum_{j=1}^n b_jY_j$, and the probability that team $A$ wins is that it lives longer, which is $G_{m,n}(\boldsymbol{a},\boldsymbol{b})$, so \eqref{eq:defG} follows.
\end{proof}

In order to prove Theorem~\ref{th:minimizer} we need several preliminary results.
Let $G_{1}, G_{2}, Z_{1}, Z_{2}$ be independent with $G_i \sim \GAMMA(u_i,1)$, $Z_i \sim \Exp(1)$, for $i=1,2$. For $u_i=0$ we define $G_i=0$ with probability 1.

 \begin{lemma} \label{le:LL}Given $a_{1}^{*},a_{2}^{*}$, set $a_1=a_1^*+\delta/u_1$ and $a_2=a_2^*-\delta/u_2$. Then
\begin{equation} \label{eq:SB}
\frac{\partial}{\partial \delta} \mathbb{P}(a_1G_1+a_2G_2\leq x) = (a_1-a_2)\frac{\partial^2}{\partial x^2} \mathbb{P}(a_1(G_1+Z_1)+a_2(G_2+Z_2)\leq x).
\end{equation}
 \end{lemma}
\begin{proof}
Let
\begin{align*}
F(x)&= \mathbb{P}(a_1G_1+a_2G_2 \le x)\\
H(x)&=\mathbb{P}(a_1G_1+a_2G_2+a_1Z_1+a_2Z_2 \le x)
\end{align*}
and let $f$ and $h$ denote the corresponding densities. Let $\mathcal{L}$ denote the Laplace transform, that is,
\[
\mathcal{L}(F)=\int_0^\infty \expo^{-tx}F(x)\diff x.
\]
Note that \eqref{eq:SB} is equivalent to
 \begin{equation}\label{eq:Lap}
\mathcal{L}\left(\frac{\partial}{\partial \delta}F(x)\right)=(a_1-a_2)\mathcal{L}\left(\frac{\partial^2}{\partial x^2}H(x)\right).
 \end{equation}
Using integration by parts we get
\[
\mathcal{L}\left(\frac{\partial^2}{\partial x^2}H(x)\right) = t\int_0^\infty \expo^{-tx}h(x)\diff x = t\ \mathbb{E}[\exp\{-t(a_1G_1+a_2G_2+a_1Z_1+a_2Z_2)\}].
\]
 For the left hand side of \eqref{eq:Lap} note that we can interchange differentiation and integration, and also that
\[
\frac{\partial}{\partial \delta}\mathcal{L}(F(x))=\mathcal{L}(F(x))\frac{\partial}{\partial \delta}\log \mathcal{L}(F(x)).
\]
Again by integration by parts we have
\[
\mathcal{L}(F(x))=\frac{1}{t}\mathcal{L}(f(x))=\frac{1}{t}\mathbb{E}[\exp\{-t(a_1G_1+a_2G_2)\}].
\]
It follows that \eqref{eq:Lap} is equivalent to
\begin{equation}\label{eq:Lap2}
\frac{1}{t}\frac{\partial}{\partial \delta}\log \mathcal{L}(f(x))= (a_1-a_2)\ t\ \mathbb{E}[\exp\{-t(a_1Z_1+a_2Z_2)\}].
 \end{equation}
Explicitly this becomes
\begin{equation}\label{eq:Lap3}
\frac{1}{t}\frac{\partial}{\partial \delta}\log[(1+a_1t)^{-u_1}(1+a_2t)^{-u_2}]= (a_1-a_2)t(1+a_1t)^{-1}(1+a_2t)^{-1}.
 \end{equation}
Using $a_1=a_1^*+\delta/u_1$, and $a_2=a_2^*-\delta/u_2$, \eqref{eq:Lap3} is verified
by a straightforward calculation.
\end{proof}
A related result to Lemma \ref{le:LL}, with a similar type of proof, appears in \citet{SzeBak:PTRF2003}.

\begin{lemma}\label{le:QQ}
Given a nonnegative vector $(a_1^*,\dots,a_m^*)$, let
\[
a_1=a_1^*+\delta/u_1,\quad a_2=a_2^*-\delta/u_2,\quad a_i=a_i^*\text{ for }3\leq i\leq m.
\]
Define
\begin{equation}\label{eq:Q}
Q(\boldsymbol{a}, \boldsymbol{u})=\sum_{i=1}^m a_i G_i - b \sum_{j=1}^n Y_j,
\end{equation}
where $(\boldsymbol{a}, \boldsymbol{u}) := (a_1,\dots, a_m, u_1,\dots, u_m)$, $G_{1}, \dots, G_{m}, Y_{1}, \dots, Y_{n}$ are independent random variables with $G_i\sim \GAMMA(u_i, 1)$, for $i=1, \dots, m$ and $Y_j \sim \Exp(1)$, for $j=1,\dots, n$.
Let $Z_i \sim \Exp(1)$, for $i=1,2$ be independent of all other variables.
Then
\begin{equation}
\label{eq:key}
 \frac{\partial}{\partial\delta} \mathbb{P}(Q(\boldsymbol{a}, \boldsymbol{u})\leq x) = (a_1-a_2)\frac{\partial^2}{\partial x^2} \mathbb{P}(Q(\boldsymbol{a}, \boldsymbol{u})+a_1 Z_1 +a_2 Z_2\leq x).
\end{equation}
\end{lemma}

\begin{proof}
Set $T=\sum_{i=3}^m a_i G_i- b \sum_{j=1}^n Y_j$. Then
\begin{equation}
\label{eq:cond}
 \frac{\partial}{\partial\delta} \mathbb{P}(Q(\boldsymbol{a}, \boldsymbol{u})\leq x|T) = (a_1-a_2)\frac{\partial^2}{\partial x^2} \mathbb{P}(Q(\boldsymbol{a}, \boldsymbol{u})+a_1 Z_1 +a_2 Z_2\leq x|T),
\end{equation}
which is equivalent to \eqref{eq:SB} with a different $x$.
Taking the expectation in \eqref{eq:cond} over $T$ yields \eqref{eq:key}.
\end{proof}

\begin{lemma}
\label{le:mode}
Let $X$ and $Y$ be independent random variables where $Y\sim\Exp(1)$
and $X$ has a density $f(x)$ such that
\begin{enumerate}[\rm (i)]
\item
$f(x)$ is continuously differentiable with a bounded derivative on $(-\infty, \infty)$,

\item
$f(x)>0$ for sufficiently small $x\in (-\infty, \infty)$,

\item
$f(x)$ is unimodal, i.e., there exists $a\in (-\infty, \infty)$ such that $f'(x)\geq 0$ if $x<a$ and $f'(x)\leq 0$ if $x>a$.
\end{enumerate}
For $\lambda> 0$, denote the density of $X+\lambda Y$ by $f_\lambda(x)$.  Then $f_\lambda(x)$ is unimodal and if $f'_\lambda(x_0)=0$ then $x_0$ is a mode of $f_\lambda$. Moreover, if $\lambda>\lambda_0> 0$, then any mode of $f_\lambda(x)$ is strictly larger than any mode of $f_{\lambda_0}(x)$.
\end{lemma}

\begin{proof}

This result is similar to \citet[Lemma~1]{SzeBak:PTRF2003}.  We provide a quick proof using variation diminishing properties of sign regular kernels \citep[see][]{Kar:SUP1968}. First, since the density of $\lambda Y$ is log-concave (a.k.a. strongly unimodal) its convolution with the unimodal $f(x)$ is also unimodal, that is, the pdf of $X+\lambda Y$ is unimodal \citep[see][]{Ibr:TVP1956, Kar:SUP1968}.

Differentiating (justified by (i)) yields
\begin{align*}
f_\lambda'(x) &= \int_0^\infty f'(x-z) \frac{1}{\lambda} \expo^{-z/\lambda} \diff z \\
&= \int_{-\infty}^x f'(z) \frac{1}{\lambda} \expo^{(z-x)/\lambda} \diff z \\
&=\frac{\expo^{-x/\lambda}}{\lambda} \int 1_{(-\infty,x)}(z) f'(z) \expo^{z/\lambda} \diff z .
\end{align*}
Suppose $f_\lambda'(x_0)=0$.  Since $f'(z) \geq 0$ for $z\leq a$, we know from the representation above that $f_\lambda'(x)>0$ if $x\leq a$, and hence $x_0>a$.  The representation also shows that the function $\expo^{x/\lambda} f_\lambda'(x)$ is nonincreasing in $x\in (a,\infty)$.  Therefore $f_\lambda'(x)\geq 0$ if $x\in (a, x_0)$ and $f_\lambda'(x)\leq 0$ if $x>x_0$.  It follows that $x_0$ is a mode of $f_\lambda(x)$.

For fixed $x$, the function $1_{(-\infty,x)}(z) f'(z)$ as a function of $z$ does not vanish (by (ii)), and has at most one sign change from positive to negative (by (iii)), and the kernel $\expo^{z/\lambda}$ is strictly reverse rule \citep[see][]{Kar:SUP1968}. It follows that $\int 1_{(-\infty,x)}(z) f'(z) \expo^{z/\lambda} \diff z $ has at most one sign change from negative to positive, as a function of $\lambda$. Thus, if for a given $x$, $f_{\lambda_0}'(x)=0$ and $\lambda>\lambda_0$, then $f_\lambda'(x) >0$, and the result follows.
\end{proof}

%
%

\begin{proof}[Proof of Theorem~\ref{th:minimizer}]
Let $Q(\boldsymbol{a}):= Q(\boldsymbol{a}, \boldsymbol{1}_m)$  as in \eqref{eq:Q}.  Consider minimizing $\mathbb{P}(Q(\boldsymbol{a})\leq 0)$ over
\begin{equation*}
\Omega = \left\{\boldsymbol{a}:\ 0\leq a_i,\ \sum_{i=1}^m a_i =m \right\}.
\end{equation*}
Since $\Omega$ is compact, and $\mathbb{P}(Q\leq 0)$ is continuous in $\boldsymbol{a}$, the minimum is attained, say, at $\boldsymbol{a}^*\in \Omega$.

\begin{claim}\label{cl:claimB}
In any minimizing point $\boldsymbol{a}^{*}$ of $\mathbb{P}(Q \le 0)$ the $a_i^*$'s take at most two distinct nonzero values.  Moreover, in the case of two distinct nonzero values, the smaller one appears only once.
\end{claim}
\begin{proof}
Assume the contrary, say $0<a_1^*\leq a_2^*<a_3^*$.  We show below in Case \ref{ca:a1lea2} that more than two distinct values are impossible by showing that $a_1^*<a_2^*$ leads to a contradiction. Similarly Case \ref{ca:a1eqa2} implies the impossibility of repetitions of the smallest of two distinct values. Let $a_1=a_1^*+\delta,\ a_2=a_2^*-\delta,\ a_i=a_i^*,\ 3\leq i\leq m$.
Then by \eqref{eq:key} we have
\begin{equation}
\label{eq:key2}
 \frac{\partial}{\partial\delta} \mathbb{P}(Q(\boldsymbol{a})\leq x) = (a_1-a_2)\frac{\partial^2}{\partial x^2} \mathbb{P}(Q(\boldsymbol{a})+a_1 Z_1 +a_2 Z_2\leq x),
\end{equation}
where $Z_1$ and $Z_2$ are i.i.d. random variables with $Z_{1}\sim\Exp(1)$,
independent of $Q$.  We can focus on $x=0$.

\begin{case}\label{ca:a1lea2}
$a_1^*<a_2^*$. Since $\delta=0$ achieves the minimum, both sides of \eqref{eq:key2} with $x=0$ vanish at $\delta=0$.  The density function  of $Q(\boldsymbol{a}^{*})+a_1^*Z_1$ is positive everywhere and is log-concave and hence unimodal.  By Lemma \ref{le:mode}, $S=Q(\boldsymbol{a}^{*})+a_1^*Z_1+a_2^* Z_2$ has a mode at zero. Following Case \ref{ca:a1eqa2} we show that this leads to a contradiction.
\end{case}

\begin{case}\label{ca:a1eqa2}
 $a_1^*=a_2^*$.  Then \eqref{eq:key2} gives
\[
\lim_{\delta\downarrow 0} \frac{\partial \mathbb{P}(Q(\boldsymbol{a})\leq 0)}{\partial \delta} =0
\]
and
\[
\left. \frac{\partial^2}{\partial\delta^2} \mathbb{P}(Q(\boldsymbol{a})\leq 0)\right|_{\delta=0} =\left. 2\lim_{\delta\to 0}\frac{\partial^2}{\partial x^2} \mathbb{P}(Q(\boldsymbol{a})+a_1 Z_1 + a_2 Z_2\leq x)\right|_{x=0}.
\]
A minimum at $\delta=0$ entails
\[
\left. \frac{\partial^2}{\partial x^2} \mathbb{P}(Q(\boldsymbol{a}^{*})+a_1^* Z_1 + a_2^* Z_2\leq x)\right|_{x=0}\geq 0,
\]
showing that $S=Q(\boldsymbol{a}^{*})+a_1^*Z_1+a_2^* Z_2$ has a mode that is nonnegative.
\end{case}

Thus $S$ has a nonnegative mode in either case.  By Lemma \ref{le:mode} and $a_2^*<a_3^*$, any mode of $Q(\boldsymbol{a}^{*})+a_1^*Z_1+a_3^* Z_2$ is strictly positive, i.e.,
\[
\left. \frac{\partial^2}{\partial x^2} \mathbb{P}(Q(\boldsymbol{a}^{*})+a_1^* Z_1 + a_3^* Z_2\leq x)\right|_{x=0}> 0.
\]
The latter expression, multiplied by $(a_1^*-a_3^*)$ is negative.  Using \eqref{eq:key2} with $a_3^*$ in place of $a_2^*$, however, this implies that $\mathbb{P}(Q(\boldsymbol{a})\leq 0)$ strictly decreases under the perturbation $(a_1^*, a_3^*)\to (a_1^*+\delta, a_3^*-\delta)$ for small $\delta>0$, which is a contradiction to the minimality at $\delta=0$. Note that the crux of the proof is in comparing two perturbations.
\end{proof}

\begin{claim}\label{cl:key}
In any minimizing point $\boldsymbol{a}^{*}$ of $\mathbb{P}(Q \le 0)$ the $a_i^*$'s are either all equal, or take only two distinct values, in which case one of them is zero.
\end{claim}

\begin{proof}
Assume the contrary, and in view of Claim 5.4, suppose we have
\[
0<a_1^*<a_2^*=\cdots = a_{k+1}^*,\ 1\leq k < m,\ a_{k+2}^*=\cdots = a_m^*=0, \quad \text{and}\quad \sum_{i=1}^m a_i^*=m.
\]
Then for some $\delta \in (0, 1/k)$,\, $a_1^*,\ldots, a_m^*$ must be of the form
\ignore{ IGNORE
\[
a_1^*=(1-k\delta)m/(k+1),\ a_2=\cdots=a_{k+1}=(1+\delta)m/(k+1),\ 0\leq \delta\leq 1/k.
\]  END IGNORE  }
\[
a_1^*=(1-k\delta)m/(k+1),\ a_2^*=\cdots=a_{k+1}^*=(1+\delta)m/(k+1),\, a_{k+2}^*=\cdots = a_m^*=0.
\]
We then have
\[
\frac{k+1}{m} Q(\boldsymbol{a})=(1- k\delta) X + (1+\delta) G - \lambda Y,\quad \lambda = \frac{b(k+1)}{m},
\]
with $X\sim \Exp(1),\ G\sim \GAMMA(k, 1),\ Y\sim \GAMMA(n, 1)$ independently.  We show that the minimum of $\mathbb{P}(Q(\boldsymbol{a})\leq 0)$ cannot be achieved in the open interval $\delta \in (0, 1/k)$, contradicting the assumption that $\boldsymbol{a}^{*}$ is a minimizer.  We have
\begin{align*}
\mathbb{P}(Q(\boldsymbol{a})\leq 0) &=\mathbb{P}\left(1+\delta (1-(k+1)B)\leq \frac{\lambda Y}{X+G}\right),
\end{align*}
where $B:= X/(X+G)$.  Note that $B$ has a $\BETA(1, k)$ distribution, $Y/(X+G)$ has a scaled $F(2n, 2(k+1))$ distribution, and $B$ and $Y/(X+G)$ are independent.  Thus
\[
\mathbb{P}(Q(\boldsymbol{a})\leq 0)=C_1\, \mathbb{E} \left[\int_{1+\delta (1-(k+1)B)}^\infty \frac{y^{n-1}}{(\lambda +y)^{n+k+1}}\, \diff y \right].
\]
where above and below, $C_i>0$ denote constants that do not depend on $\delta$, and $D_i(\delta)>0$ denote functions of $\delta\in (0, 1/k)$, and both may depend on other constants such as $\lambda, k$, etc.  It follows that
\begin{align}
\nonumber
\frac{\partial \mathbb{P}(Q(\boldsymbol{a})\leq 0)}{\partial \delta} &=-C_1 \, \mathbb{E}  \left[(1-(k+1)B) \frac{(1+\delta (1-(k+1)B))^{n-1}}{(\lambda +1+\delta(1-(k+1)B))^{n+k+1}}\right]\\
\label{eq:ginte}
&= -C_2 \int_{-k}^1 x (x+k)^{k-1} \frac{(1+\delta x)^{n-1}}{(\lambda +1+\delta x)^{n+k+1}}\, \diff x\\
\label{eq:gdelta}
&= -D_1(\delta) g(\delta),
\end{align}
where
\begin{align*}
g(\delta) &:= \int_1^{p} \left[(\lambda +1-\delta k)(y-1)-\delta k \lambda y\right] y^{n-1} (y-1)^{k-1}\, \diff y,\\
\nonumber
p = p(\delta) &:= \frac{(1+\delta)(\lambda +1-\delta k)}{(\lambda +1+\delta)(1-\delta k)},
\end{align*}
and \eqref{eq:gdelta} uses the change of variables
\[
y=\frac{(1+\delta x)(\lambda +1-\delta k)}{(\lambda +1+\delta x)(1-\delta k)}.
\]
Using the closed form integral
\[
\int_1^p\left[ky + n(y-1)\right] y^{n-1} (y-1)^{k-1}\, \diff y = p^n(p-1)^k
\]
we get
\begin{align}
\nonumber
g'(\delta) &= \frac{\lambda \delta (\lambda +1-\delta k)}{\lambda +1 +\delta} p^{n-1} (p-1)^{k-1}p'(\delta) + \int_1^p k(1-(\lambda +1)y) y^{n-1} (y-1)^{k-1}\, \diff y\\
\nonumber
&=\frac{\lambda \delta (\lambda +1-\delta k)}{\lambda +1 +\delta} p^{n-1} (p-1)^{k-1}p'(\delta) + \frac{(\lambda n -k) g(\delta) - \lambda(\lambda +1) p^n (p-1)^k}{\lambda +1-\delta k + \lambda n \delta}\\
\nonumber
&=D_2(\delta) \left[k(\lambda n -k) \delta^2 + (\lambda +1)(k-1)\delta +(\lambda +1)(\lambda (n-1) -k-2)\right]\\
\label{eq:gprime}
&\quad\quad  + \frac{(\lambda n -k) g(\delta)}{\lambda +1-\delta k + \lambda n \delta}.
\end{align}
Specifically
\[
D_2(\delta)=\frac{\lambda \delta p^{n}(p-1)^{k}}{(1 + \delta)(\lambda+1+\delta) (\lambda+1-\delta k+\lambda n\delta)}.
\]

It is helpful to determine the sign of $g(\delta)$ for small $\delta>0$ and large $\delta<1/k$.  Let us denote the integral in \eqref{eq:ginte} by $\tilde{g}(\delta)$, which has the same sign as $g(\delta)$ for $\delta\in (0, 1/k)$.  A Taylor expansion yields
\begin{align*}
\tilde{g}(\delta) &=\int_{-k}^1 \left[\frac{x(x+k)^{k-1}}{(\lambda+1)^{n+k+1}} + \frac{(\lambda(n-1)-k-2) \delta} {(\lambda+1)^{n+k+2}} x^2(x+k)^{k-1} \right]\, \diff x + o(\delta)\\
& = C_3 (\lambda(n-1)-k-2)\delta +o(\delta),\quad \text{as}\ \delta\downarrow 0.
\end{align*}
By direct calculation,
\[
\tilde{g}(1/k)= C_4 (\lambda (n-1)-k-1).
\]

We distinguish three cases:
\begin{enumerate}[(i)]
\item\label{ca:ineq-1}
$\lambda (n-1) > k+2$.  Then $\tilde{g}(\delta)>0$ and hence $g(\delta)>0$ for sufficiently small $\delta>0$.  Moreover, by \eqref{eq:gprime}, $g'(\delta)> D_3(\delta) g(\delta),\ \delta\in (0,1/k)$.  It follows that $g(\delta)>0$ for all $\delta\in (0, 1/k)$, i.e., $\mathbb{P}(Q(\boldsymbol{a})\leq 0)$ decreases in $\delta\in [0, 1/k]$.  The same holds in the boundary case $\lambda (n-1) = k+2$.

\item\label{ca:ineq-2}
$k+1< \lambda (n-1)< k+2$.  Then $g(\delta)<0$ for sufficiently small $\delta>0$, and $g(\delta)>0$ for sufficiently large $\delta <1/k$.  If the minimum of $\mathbb{P}(Q(\boldsymbol{a})\leq 0)$ is achieved at $\delta^*\in (0, 1/k)$, then $g(\delta^*)=0\geq g'(\delta^*)$, and $g(\delta)$ has at least one root in $(0, \delta^*)$, say $\delta^{**},$ such that $g'(\delta^{**})\geq 0$.  This contradicts \eqref{eq:gprime}, however, because the term in square brackets strictly increases in $\delta$.

\item\label{ca:ineq-3}
$\lambda (n-1)< k+1$.  Then $g(\delta)<0$ for both sufficiently large $\delta <1/k$ and sufficiently small $\delta >0$.  Suppose $g(\delta^*)> 0$ for some $\delta^*\in (0, 1/k)$.  If $\lambda n>k$ then a contradiction results as in Case (ii).  Otherwise the term in square brackets in \eqref{eq:gprime} is no more than
\[
(\lambda +1)(k-1)k^{-1} + (\lambda +1)(\lambda (n-1)-k-2)<0.
\]
Thus any $\delta\in (0, 1/k)$ such that $g(\delta)=0$ entails $g'(\delta)<0$.  This is impossible as $g(\delta)$ cannot cross zero from above without first doing so from below.  Hence $g(\delta)\leq 0$, i.e., $\mathbb{P}(Q(\boldsymbol{a})\leq 0)$ increases in $\delta\in [0, 1/k]$.  The same holds in the boundary case $\lambda (n-1)=k+1$. \qedhere
\end{enumerate}
\end{proof}

We now prove the three statements of Theorem~\ref{th:minimizer}.
\begin{itemize}
\item[\eqref{it:th:minimizer-a}]
This is an immediate consequence of Claim~\ref{cl:key}.

\item[\eqref{it:th:minimizer-b}]
Let $h(k)=\mathbb{P}(Q(\boldsymbol{a})\leq 0)$ with
\[
a_1=\dots= a_k=\frac{m}{k},\ 1\leq k\leq m,\quad\text{and}\quad a_{k+1}=\cdots =a_m=0.
\]
Comparing $\mathbb{P}(Q(\boldsymbol{a})\leq 0)$ in Case~\eqref{ca:ineq-3} of the proof of Claim (\ref{cl:key}) at $\delta=0$ and $\delta=1/k$,
we see that if $m \ge b(n-1)$, i.e.,
\[
\frac{b(k+1)(n-1)}{m} \le k+1,
\]
then $h(k+1) < h(k),\ 1\leq k <m$.  Thus $h(k)$ achieves its minimum at $k=m$.

\item[\eqref{it:th:minimizer-c}]
Suppose now $m < b(n-1)$.  According to Case~\eqref{ca:ineq-1}, if $b(k+1)(n-1)/m \geq k+2$, i.e.,
\begin{equation}\label{eq:last}
k+1\geq \frac{m}{(b(n-1)-m)},
\end{equation}
then $h(k+1)>h(k)$.  In particular, \eqref{eq:last} holds for all $k$ if $m\leq 2b(n-1)/3$, which yields $h(m)>\cdots > h(2)>h(1)$, i.e., $h(k)$ is minimized at $k=1$.  In general $h(k)$ is minimized at some $k\leq \lceil m/(b(n-1)-m) -1 \rceil$.    \qedhere
\end{itemize}
\end{proof}

\begin{proof}[Proof of Theorem~\ref{th:Nash}]
\begin{itemize}
\item[\eqref{it:th:Nash-a}]
Using Proposition~\ref{pr:kaminsky} and Theorem~\ref{th:minimizer}\eqref{it:th:minimizer-a}\eqref{it:th:minimizer-b}, once all the $a_{i}$ are multiplied by a factor $c_{A}/m$, we can prove that there exists a Nash equilibrium that satisfies \eqref{eq:equila} and \eqref{eq:equilb}, which we denote as $(\boldsymbol{a}^*, \boldsymbol{b}^*)$.  Assume $(\widetilde{\boldsymbol{a}}, \widetilde{\boldsymbol{b}})$ is another equilibrium.  Because the game is zero-sum, we have
\[
H_{m.,n}(\widetilde{\boldsymbol{a}}, \widetilde{\boldsymbol{b}}) \ge H_{m,n}(\boldsymbol{a}^*, \widetilde{\boldsymbol{b}}) \geq H_{m,n} (\boldsymbol{a}^*, \boldsymbol{b}^*)
\]
and
\[
H_{m,n}(\widetilde{\boldsymbol{a}}, \widetilde{\boldsymbol{b}}) \le H_{m,n}(\widetilde{\boldsymbol{a}}, \boldsymbol{b}^*) \leq H_{m,n}(\boldsymbol{a}^*, \boldsymbol{b}^*).
\]
Thus equalities must all hold.  Since $\boldsymbol{b}^*$ (equal allocation) is the unique optimal response to $\boldsymbol{a}^*$, for the equality to hold in $H_{m,n}(\boldsymbol{a}^*, \widetilde{\boldsymbol{b}}) \geq H_{m,n} (\boldsymbol{a}^*, \boldsymbol{b}^*)$ we must have $\widetilde{\boldsymbol{b}} = \boldsymbol{b}^*$.  Similarly, for the equality to hold in $H_{m,n}(\widetilde{\boldsymbol{a}}, \boldsymbol{b}^*) \leq H_{m,n}(\boldsymbol{a}^*, \boldsymbol{b}^*)$, $\widetilde{\boldsymbol{a}}$ must be of the form \eqref{eq:equila}.  Thus all pure equilibria satisfy \eqref{eq:equila} and \eqref{eq:equilb}.

\item[\eqref{it:th:Nash-b}]
Theorem~\ref{th:minimizer}\eqref{it:th:minimizer-b} guarantees that if $a_{1}^{*}=\dots=a_{m}^{*} = c_{A}/m$ and $b_{1}^{*}=\dots=b_{n}^{*}=c_{B}/n$, then
$\boldsymbol{a}^{*}$ is the unique best response to $\boldsymbol{b}^{*}$ and vice versa. This proves that $(\boldsymbol{a}^{*}, \boldsymbol{b}^{*})$ is a Nash equilibrium of the game.  This equilibrium is unique by the argument in part \eqref{it:th:Nash-a}.

\item[\eqref{it:th:Nash-c}]
Theorem~\ref{th:minimizer}\eqref{it:th:minimizer-c} guarantees that if $a_{i}^{*}=c_{A}$ for some $i \in \{1,\dots,m\}$ and $a_{j}^{*}=0$ for all $j \ne i$, and
 $b_{1}^{*}=\dots=b_{n}^{*}=c_{B}/n$, then $\boldsymbol{a}^{*}$ is a best response to $\boldsymbol{b}^{*}$ and Theorem~\ref{th:minimizer}\eqref{it:th:minimizer-b} guarantees that $\boldsymbol{b}^{*}$ is the unique best response to $\boldsymbol{a}^{*}$. This proves that $(\boldsymbol{a}^{*}, \boldsymbol{b}^{*})$ is a Nash equilibrium of the game. Again the argument used in part \eqref{it:th:Nash-a} shows that all Nash equilibria are of this form.

\item[\eqref{it:th:Nash-d}]
Suppose team $A$ allocates its strength equally among $r$ players, and team $B$ adopts the optimal strategy of equal allocation among all $n$ players.
Then, as $n\to \infty$, the winning probability for team $A$ approaches $f(r):= \mathbb{P}(c_A G_r > r c_B)$, where $G_r$ is a $\GAMMA(r, 1)$ random variable.  Letting $\beta:=c_B/c_A$, we get
\begin{align*}
f(r)-f(r+1) &= \int_{r\beta}^\infty \frac{r x^{r-1} \expo^{-x}}{\Gamma(r+1)} \diff x -\int_{(r+1)\beta}^\infty \frac{x^r \expo^{-x}}{\Gamma(r+1)} \diff x\\
& = \frac{1}{\Gamma(r+1)} \left( -(r\beta)^r \expo^{-r\beta} + \int_{r\beta}^{(r+1)\beta} x^r \expo^{-x} \diff x\right)\\
& = \frac{(r\beta)^{r} \expo^{-r \beta}}{\Gamma(r+1)}\left[\int_0^1 \left(1+\frac{y}{r}\right)^r \expo^{-y\beta} \beta \diff y - 1\right],
\end{align*}
where we have integrated by parts in the second equality and changed the variables $y= x/\beta -r$ in the third.  The integral inside the square brackets obviously increases in $r$.  Hence $f(r)> f(r+1)$ implies $f(r+1)> f(r+2)> \cdots >f(m)$.  Moreover, if $\beta = c_B/c_A > t_0$ then $f(1) > f(2)$ by direct calculation.  In this case $f(r)$ is maximized at $r=1$ and $r=1$ is the optimal strategy for team $A$ in the large $n$ limit. \qedhere
\end{itemize}
\end{proof}

\begin{proof}[Proof of Theorem~\ref{th:value}]
\begin{itemize}
\item[\eqref{it:th:value-a}]
Using Theorem~\ref{th:Nash}\eqref{it:th:Nash-a} we know that for some $1 \le r \le m$ and some permutation $\pi$ we have $a_{\pi(1)}^{*} = \dots = a_{\pi(r)}^{*} = c_{A}/r$, $a_{\pi(r+1)} = \dots = a_{\pi(m)} = 0$, and $b_1^*=\cdots =b_n^* = c_{B}/n$. Hence
\[
\sum_{i=1}^m a^*_{i} X_{i} \sim\GAMMA(r, r/c_{A}), \quad
\sum_{j=1}^n b^*_{j} Y_{j} \sim\GAMMA(n, n/c_{B}).
\]
Therefore, \citep[see, e.g,][]{CooNad:BJ2006, Coo:UTMD2008}
\begin{equation}\label{eq:Pincompletebeta}
\mathbb{P}\left(\sum_{i=1}^m a_i^* X_i > \sum_{j=1}^n b_j^* Y_j \right) = 1-I\left(\frac{rc_{B}}{rc_{B}+nc_{A}}, r,n \right),
\end{equation}
where $I$ is the regularized  incomplete beta function defined in \eqref{eq:incompletebeta}.

\item[\eqref{it:th:value-b}]
By Theorem~\ref{th:Nash}\eqref{it:th:Nash-b} in this case $r=m$.

\item[\eqref{it:th:value-c}]
By Theorem~\ref{th:Nash}\eqref{it:th:Nash-c} in this case $r=1$. \qedhere
\end{itemize}
\end{proof}

%
%
%
%

\section{Monotonicity results}\label{se:monotonicity}

\subsection{Monotonicity of the value}

\begin{center}
FIGURE~\ref{fi:plotcaeqcb}  ABOUT HERE
\end{center}

%
%
%

We mention the following consequence of Theorem~\ref{th:Nash} (see Figure~\ref{fi:plotcaeqcb}).

\begin{corollary}\label{co:equalblotto}
In the game $\mathcal{G}(m,n,c_{A},c_{B})$, if the two teams have equal strength (i.e., $c_{A}=c_{B}$), then the value is positive if $m>n$, namely, the team with more players has an advantage over the other team. Moreover, the value of the game is increasing in $m$ and decreasing in $n$.
\end{corollary}

\begin{proof}
The team with more players always has the option of not using them all.  Therefore it cannot be worse off than the team with fewer players.  However, since equal allocation is the unique best response, using them all is strictly better. The same argument proves the monotonicity in $m$ and $n$. Note that directly verifying this from the properties of the incomplete beta function appears nontrivial.
\end{proof}

\begin{center}
FIGURE~\ref{fi:plotcasimcb}  ABOUT HERE
\end{center}

Figure~\ref{fi:plotcasimcb} shows an interesting implication of Theorem~\ref{th:value}: team $B$ may be at a disadvantage even if $c_{A} < c_{B}$, and this happens if the number $n$ of its gladiators is much smaller than the number $m$ of gladiators in $A$. As the relative difference in strength between the two teams increases, it takes a larger number of gladiators to compensate for the lower strength.

\begin{center}
FIGURES~\ref{fi:plotvariouscb} AND  \ref{fi:plotvariousn} ABOUT HERE
\end{center}

As Figures~\ref{fi:plotvariouscb} and  \ref{fi:plotvariousn} show, if condition \eqref{eq:inequalitycasmallercb} holds, then team $A$ is at a strong disadvantage. The disadvantage increases with the total strength $c_{B}$ and the number $n$ of gladiators of team $B$. The number $m$ of gladiators of team $A$ is totally irrelevant, since, in equilibrium, the whole strength $c_{A}$ is assigned to only one gladiator.

\subsection{Related probability inequalities}\label{suse:additional}

If $X_{1}, \dots, X_{m}$, and $Y_{1},\dots,Y_{n}$ are i.i.d. random variables with $X_{1}\sim \Exp(1)$, and
\[
\bar{X} = \frac{1}{m}\sum_{i=1}^m X_i,\quad \bar{Y} = \frac{1}{n}\sum_{j=1}^n Y_j, \quad Z = \frac{m\bar{X}}{m\bar{X}+n\bar{Y}},
\]
then $Z$ has a $\BETA(m,n)$ distribution.  Hence
\begin{equation}\label{eq:conbeta}
\mathbb{P}(\bar{X} < \bar{Y})
=\mathbb{P}\left(Z < \frac{m}{m+n}\right)
= I\left(\frac{m}{m+n},m,n \right).
\end{equation}
For $m > n$, by Corollary~\ref{co:equalblotto}, we have
\begin{equation}\label{eq:means}
\mathbb{P}(\bar{X} < \bar{Y})<\frac{1}{2}.
\end{equation}
Since $\mathbb{E}[Z]=m/(m+n)$, \eqref{eq:means} is equivalent to $\mathbb{P}\left(Z < \mathbb{E}[Z]\right) < 1/2$, that is, $\mathbb{E}[Z] < \Med[Z]$.
This is a well known mean-median inequality for beta distributions \citep[see][]{GroMee:AS1977}.

The inequality \eqref{eq:means} has the following interesting statistical implication.
If two statisticians estimate the mean of
exponential variables, and use the sample mean as their unbiased
estimate, then the statistician with the larger sample tends to have
a larger (unbiased) estimate. If the two of them bet on
who has a larger estimate, the one with the larger sample tends to
win. For normal variables, or any symmetric variables, this clearly
cannot happen and  $\mathbb{P}(\bar{X} < \bar{Y})=1/2$.

Suppose now that the two statisticians share the first $n$
variables, that is, for $i=1,\dots,n$ we have $X_i=Y_i$, and the
remaining variables $X_{n+1},\dots,X_m$ are independent of the previous
ones. Then
\begin{align}\label{eq:medcom}
\mathbb{P}(\bar X < \bar Y)&=\mathbb{P}\left(\frac{1}{m}\left[\sum_{j=1}^n Y_j + \sum_{i=n+1}^m X_i\right] <
\frac{1}{n}\sum_{j=1}^n Y_j\right) \nonumber\\
&=\mathbb{P}\left(\frac{1}{m-n}\sum_{i=n+1}^m X_i
< \frac{1}{n}\sum_{j=1}^n Y_j\right).
\end{align}
By \eqref{eq:means} the last expression in \eqref{eq:medcom} is less than $1/2$ if
and only if $m-n >n$, that is,  $m>2n$. It equals $1/2$ if $m=2n$, and it is
larger than $1/2$ if $m<2n$, in which case \eqref{eq:means} is reversed. Thus in the
bet between the statisticians, if most of the variables are  in
common, the odds are against the one with the larger sample, contrary to the previous situation. This was noted by Abram Kagan.

Our main results can be presented in terms of various other distributional inequalities or monotonicity.
Using \eqref{eq:binombeta} and Corollary~\ref{co:equalblotto} we obtain further results that cannot easily be proved more directly.
We say that $X\sim\GAMMA(\alpha,\beta)$ if $X$ has a density
\[
f(x) = \frac{\beta^{\alpha}}{\Gamma(\alpha)}\expo^{-\beta x}x^{\alpha-1},\quad x>0.
\]

\begin{corollary} \label{co:betabin}
For $m,n$ integers the following properties hold:
\begin{enumerate}[\rm (a)]
\item\label{it:co:betabin-a}
The function
\[
I\left(\frac{m}{m+n},m,n\right)
\]
is decreasing in $m$ for fixed $n$, and increasing
in $n$ for fixed $m$.

\item\label{it:co:betabin-b}
Let $T \sim \BINOMIAL(m+n-1, m/(m+n))$. Then $\mathbb{P}(T
\ge m)$ is decreasing in $m$ and increasing in $n$.

\item\label{it:co:betabin-c}
Let $S \sim
\Poisson(m)$. Then $\mathbb{P}(S \ge m)$ is decreasing in $m$.

\item\label{it:co:betabin-d}
Let $R\sim \GAMMA(m,1)$. Then $\mathbb{P}(R \le m)$ is decreasing in $m$.
\end{enumerate}
\end{corollary}

\begin{proof}
\begin{itemize}
\item[\eqref{it:co:betabin-a}]
is a restatement of the last part of  Corollary~\ref{co:equalblotto}.

\item[\eqref{it:co:betabin-b}] follows from \eqref{it:co:betabin-a} and \eqref{eq:binombeta}.

\item[\eqref{it:co:betabin-c}] follows from \eqref{it:co:betabin-b} by letting $n \rightarrow \infty$.

\item[\eqref{it:co:betabin-d}]
follows from \eqref{it:co:betabin-c} and the identity
\[
\mathbb{P}(S \ge m) = \frac{1}{\Gamma(m)}\int_0^m \expo^{-t}\,t^{m-1} \diff t. \qedhere
\]
\end{itemize}
\end{proof}

We say that a random variable $Q\sim\Geom(p)$  if   $\mathbb{P}(Q_{1}=k)=(1-p)^k p$, $k=0,1,2,\dots$.

\begin{proposition}\label{pr:geometric}
Let $Q_{1}, \dots, Q_m$ be independent random variables such that $Q_{i}\sim\Geom(1/(1+a_{i}))$.  Define $Q =\sum_{i=1}^m Q_i$.
\begin{enumerate}[\rm (a)]
\item\label{it:co:Negbin-a}
We have
\begin{equation}\label{eq:geo0}
1-G_{m,n}(\boldsymbol{a},\boldsymbol{1}_n)=\mathbb{P}\left(Q \le n-1\right),
\end{equation}
where $\boldsymbol{a}=(a_1, \ldots, a_m)$ and $\boldsymbol{1}_n$ denotes the $n$-dimensional vector of ones.
\item\label{it:co:Negbin-b}
If $\sum_{i=1}^m a_i=n$,
then the probability in \eqref{eq:geo0} is minimized when all $a_i$'s are equal. In this case $Q_i$ are i.i.d. and $Q$ has a negative binomial distribution.
\item\label{it:co:Negbin-c}
If $\mathbb{E}[Q] = m$, then $\mathbb{E}[Q] > \Med[Q]$.
\end{enumerate}
\end{proposition}

\begin{proof}
\begin{itemize}
\item[\eqref{it:co:Negbin-a}]
The relation \eqref{eq:geo0} can be explained directly: team $A$ loses if all its gladiators together
defeat at most $n-1$ opponents. Gladiator $i$ from team $A$ defeats a geometric random number, $Q_i$, of gladiators
of strength 1 from team $B$ since he fights until he loses, and he loses a fight with probability $1/(1+a_{i})$. Thus
if $\sum_{i=1}^m Q_i \le n-1$, then team $A$ defeats  at most $n-1$ gladiators altogether, and loses.

\item[\eqref{it:co:Negbin-b}]
This follows directly from Theorem \ref{th:minimizer}.

\item[\eqref{it:co:Negbin-c}]
Note that $\mathbb{E}[Q] = \sum_{i=1}^m a_i$.
Letting $n=m$, and using \eqref{eq:geo0} and part \eqref{it:co:Negbin-b}, we conclude that $\mathbb{P}(Q\le n-1)\geq 1-G_{m,n}(\boldsymbol{1}_m,\boldsymbol{1}_n)=1/2$.  We obtain
$\mathbb{P}(Q \le \mathbb{E}[Q])=\mathbb{P}(Q \le n)>1/2$, and therefore $\mathbb{E}[Q] > \Med[Q]$. \qedhere
\end{itemize}
\end{proof}

 \section{Comments and extensions}\label{se:extensions}

The probability in \eqref{eq:aoveraplusb} is a particular example of \emph{contest success function}\footnote{\citet{Hir:PC1989} calls it technology of conflict}. The following more general class was considered by
\citet{Tul:TAMUP1980}
with the purpose of studying efficient rent seeking:
\begin{equation}\label{eq:csf}
h_{\gamma}(a,b) = \frac{a^{\gamma}}{a^{\gamma}+b^{\gamma}}, \quad \gamma > 0.
\end{equation}
These functions have been studied, axiomatized, and widely used in different fields \citep[see, e.g.,][and many others]{Ska:ET1996, Szy:JEL2003, CorDah:ET2010}. The reader is referred to
\citet{Cor:RED2007}, \citet{GarSka:HDE2007}, \citet{Kon:OUP2009}
for surveys on this topic.

In \eqref{eq:csf}, when $\gamma \to \infty$, then
\[
h_{\gamma}(a,b) \to h_{\infty}(a,b):=
\begin{cases}
1 & \text{if $a>b$}, \\
\frac{1}{2} & \text{if $a=b$}, \\
0 & \text{if $a<b$}.
\end{cases}
\]
This case corresponds to a classical Colonel Blotto situation where the stronger gladiator always wins. If the contest success function $h_{\infty}$ is used in our game, then any equilibrium strategy for the stronger team assigns the whole strength to one single gladiator, and, for $c_{A} < c_{B}$,  team $A$  loses with probability one and value of the game is $-1/2$.

In \eqref{eq:csf}, when $\gamma \to 0$, then
\[
h_{\gamma}(a,b) \to h_{0}(a,b):=
\begin{cases}
1 & \text{if $a>b=0$}, \\
\frac{1}{2} & \text{if $a, b > 0$}, \\
0 & \text{if $0=a<b$}.
\end{cases}
\]
When $h_{0}$ is used as a contest success function in our game, then any equilibrium strategy assigns positive strength to every gladiator, therefore  in each  fight either gladiator wins with probability $1/2$ and the game reduces to one with two teams of $m$ and $n$ gladiators respectively, all having equal power. Then, using  \eqref{eq:defG}, and \eqref{eq:Pincompletebeta} we see that the probability that team $A$ wins is equal to 
\[
G_{m,n}(\boldsymbol{1},\boldsymbol{1})=1- I\left(\frac{1}{2}, m,n\right). 
\]
If $a_{1} = \dots = a_{m} = 1$, then in  \eqref{eq:geo0} the random variable $Q$ is negative binomial. Hence it is easy to see that 
\[
G_{m,n}(\boldsymbol{1},\boldsymbol{1})=\sum_{j=0}^{m-1} \left(\frac{1}{2}\right)^{n+j} \binom{n+j-1}{j},\]
and the value of the game is obtained by subtracting $1/2$.

If the extreme cases $\gamma = 0$ and $\gamma = \infty$ are easy to analyze, and the case $\gamma = 1$ required hard calculations, the remaining  cases, i.e., $\gamma \not \in \{0, 1, \infty\}$ look prohibitive in our model. They were considered in easier to deal frameworks by some authors.
For instance, in a context of rent-seeking, when a contest success function of type \eqref{eq:csf} is used,
\citet{AlcDah:JPE2010} show that for $\gamma \ge 2$ the structure of the equilibrium is always the same.

\citet{Fri:OR1958} and
\citet{Rob:ANUmimeo2005} consider the case $\gamma=1$ in a static simultaneous battle context similar to  the classical Colonel Blotto model and show that the equilibrium strategies for both players involve splitting strength evenly across all the battlefields.
\citet{Rob:ET2006} considers the case $\gamma = \infty$ and shows that the equilibrium mixed strategy of the stronger player stochastically assigns positive strength to each battlefield, whereas the one of the weaker player gives zero strength to some randomly selected battlefields and randomly distributes the strength among the remaining fields.
These results bear some analogy with the structure of the equilibrium in our game.

\citet{TanShoLin:AI2010} consider contest games where the strengths of players are exogenously given and coaches simultaneously choose the order of players and then players with the corresponding position fight.
This model was used by
\citet{Ara:mimeo2009}.

\section*{Acknowledgments}
The gladiator game of \citet{KamLukNel:AJS1984} was pointed to us by Gil Ben Zvi. We are grateful to Sergiu Hart, Pierpaolo Brutti,  Abram Kagan, Paolo Giulietti, Chris Peterson, and Andreas Hefti  for their interest and excellent advice. We thank two referees and an associate editor for their insightful comments.

\bibliographystyle{artbibst}
\bibliography{bibgladiator}

\newpage

\section*{Figures}

\begin{figure}[h]
\centering
\includegraphics[width=15cm, height=7cm]{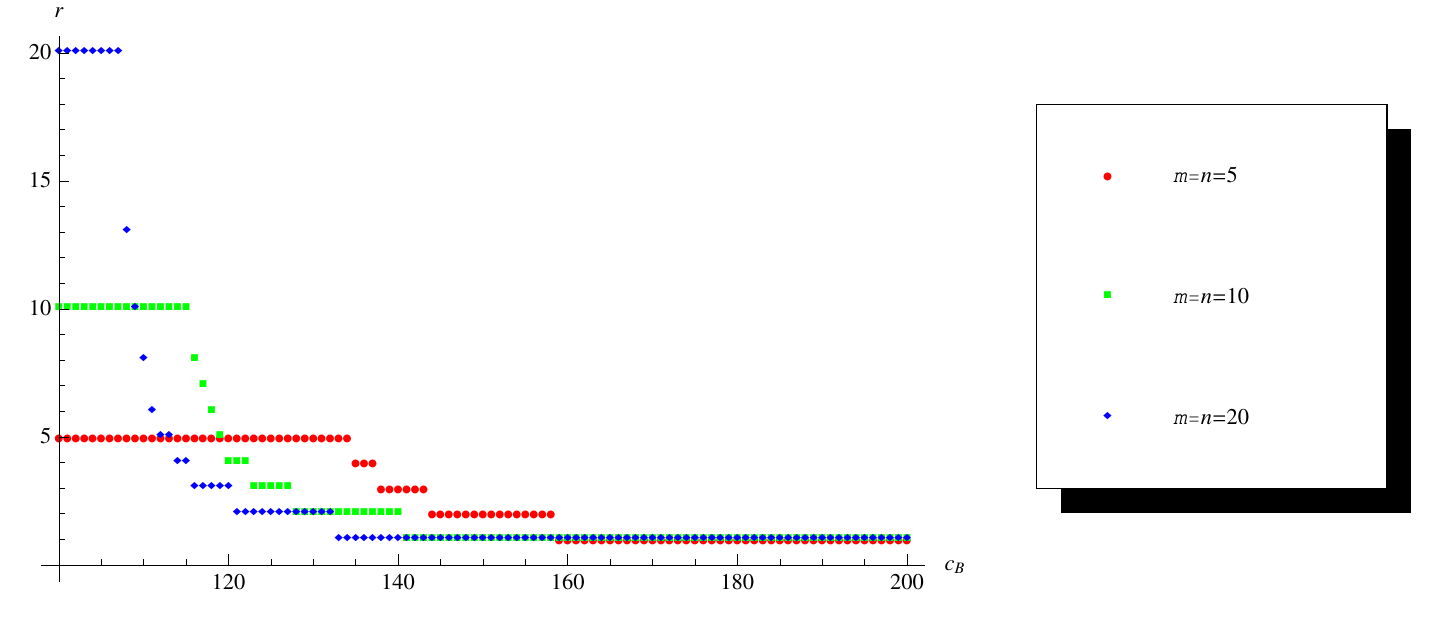}
~\vspace{0cm} \caption{\label{fi:plotmeqn} Number of positive $a_{i}^{*}$ as a function of $c_{B}$ for $c_{A}=100$ and various $m=n$.}
\end{figure}

\begin{figure}[h]
\centering
\includegraphics[width=15cm, height=7cm]{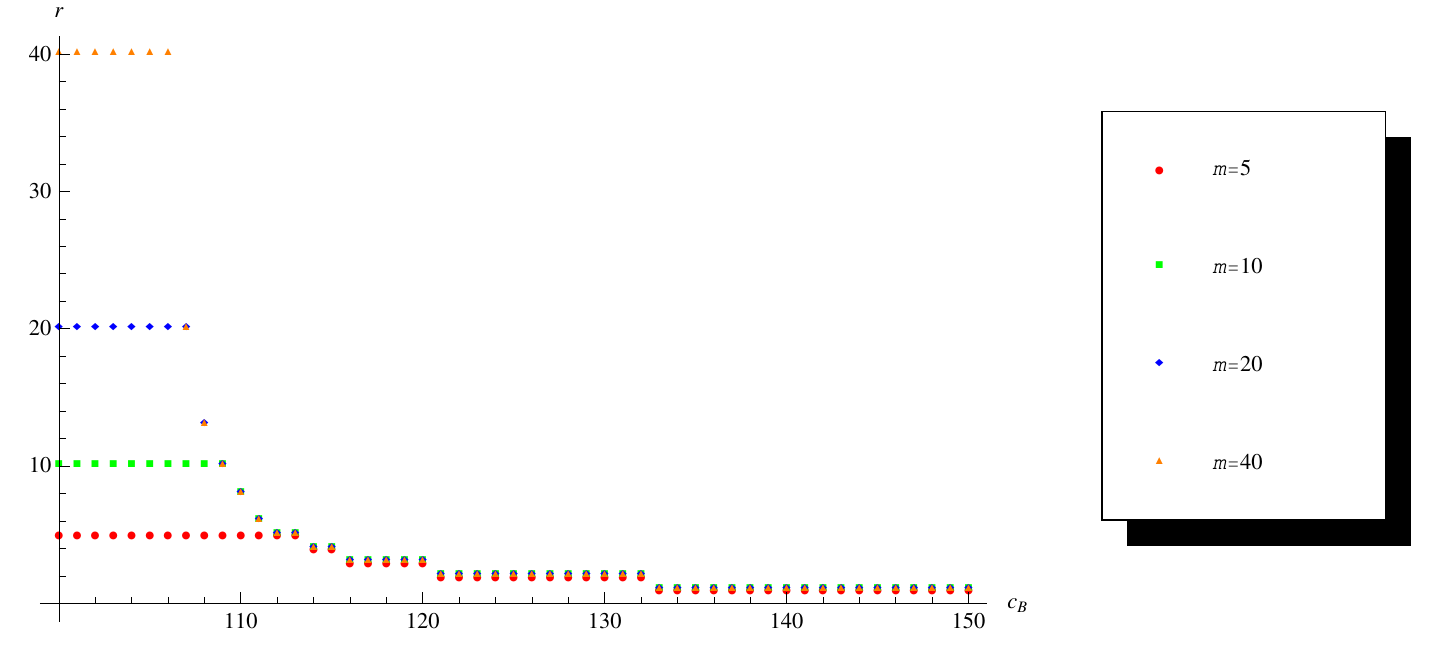}
~\vspace{0cm} \caption{\label{fi:plotfixedn} Number of positive $a_{i}^{*}$ as a function of $c_{B}$ for $c_{A}=100$, $n=20$, and various $m$.}
\end{figure}

\begin{figure}[h]
\centering
\includegraphics[width=15cm, height=7cm]{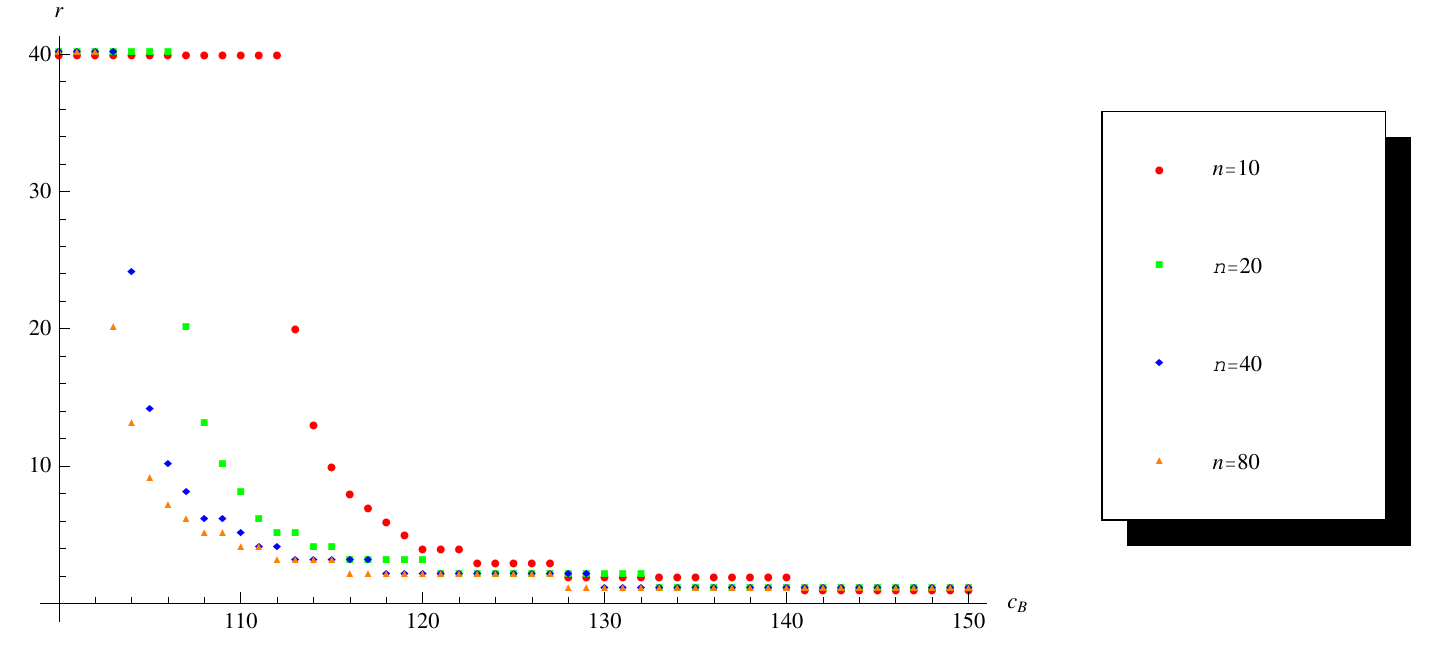}
~\vspace{0cm} \caption{\label{fi:plotfixedm} Number of positive $a_{i}^{*}$ as a function of $c_{B}$ for $c_{A}=100$, $m=40$, and various $n$.}
\end{figure}

\begin{figure}[h]
\centering
\includegraphics[width=15cm, height=7cm]{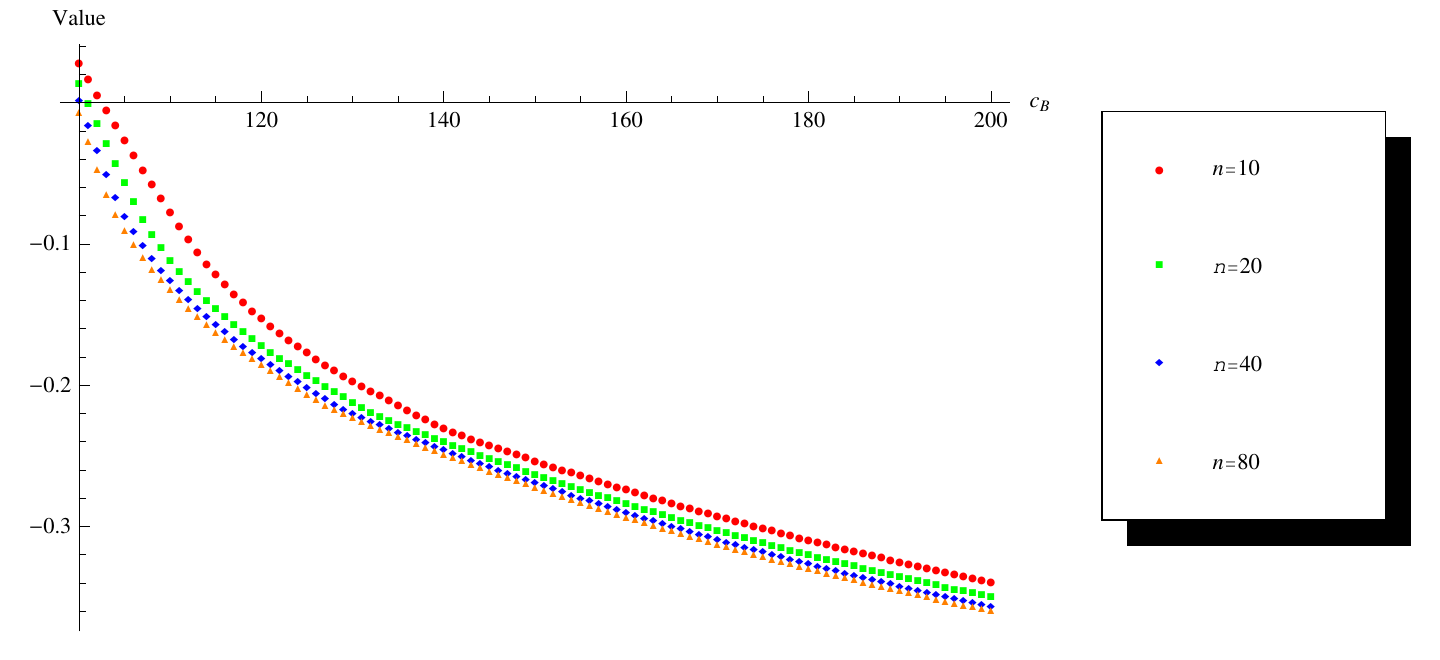}
~\vspace{0cm} \caption{\label{fi:plotvariousr} Value of $\mathcal{G}$ as a function of $c_{B}\in[100,200]$  for $c_{A}=100$, $m=40$, and various $n$.}
\end{figure}

\begin{figure}[h]
\centering
\includegraphics[width=15cm, height=7cm]{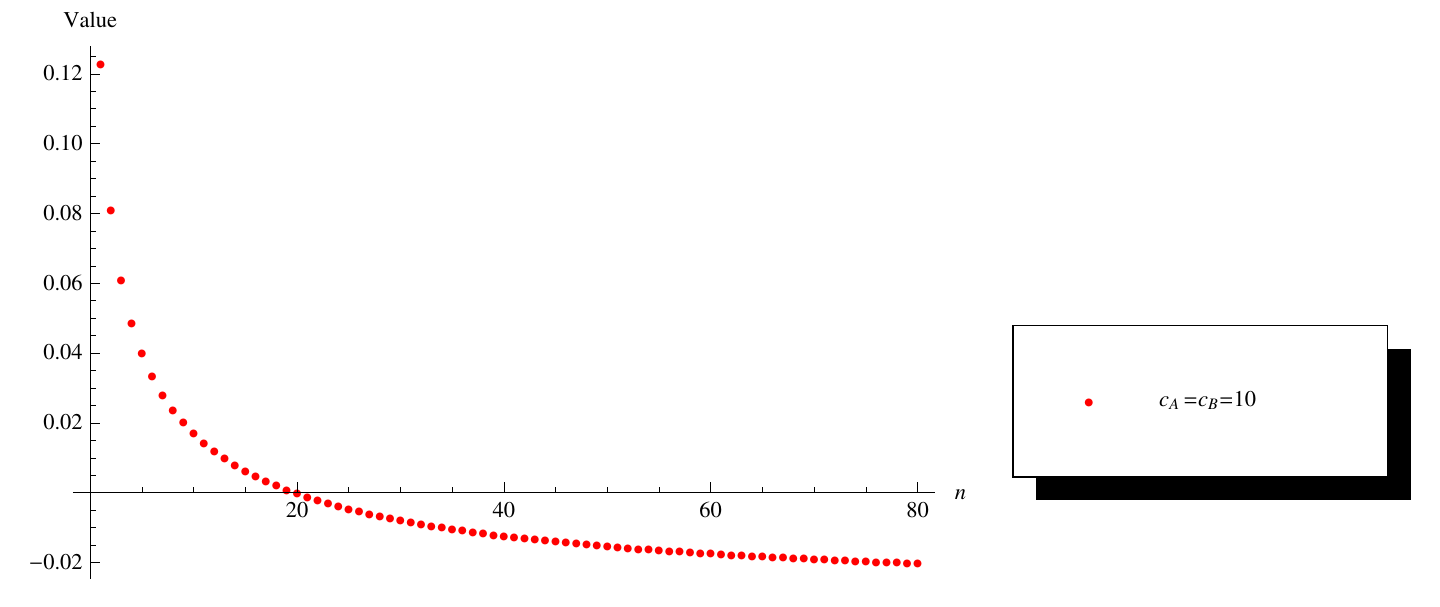}
~\vspace{0cm} \caption{\label{fi:plotcaeqcb} Value of $\mathcal{G}$ as a function of $n$ for $m=20$ and  $(c_{A}=c_{B})$.}
\end{figure}

\begin{figure}[h]
\centering
\includegraphics[width=15cm, height=7cm]{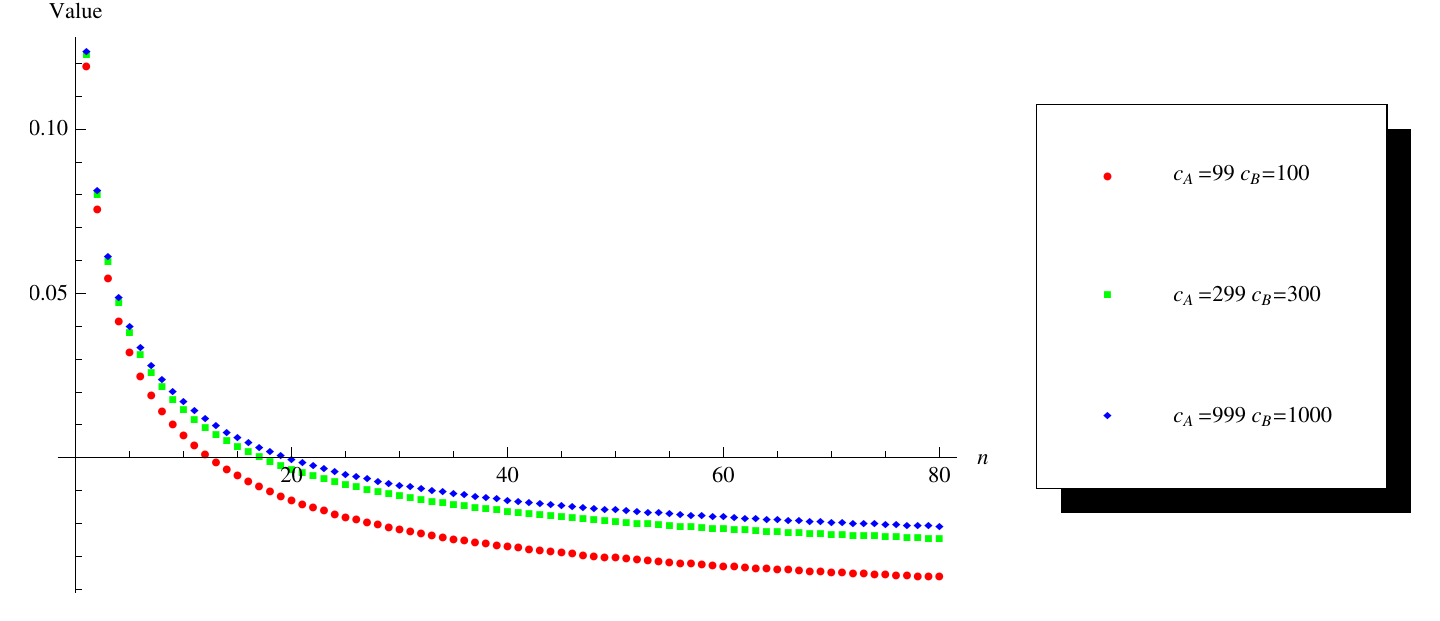}
~\vspace{0cm} \caption{\label{fi:plotcasimcb} Value of $\mathcal{G}$ as a function of $n$ for $m=20$ and different pairs $(c_{A},c_{B})$.}
\end{figure}

\begin{figure}[h]
\centering
\includegraphics[width=15cm, height=7cm]{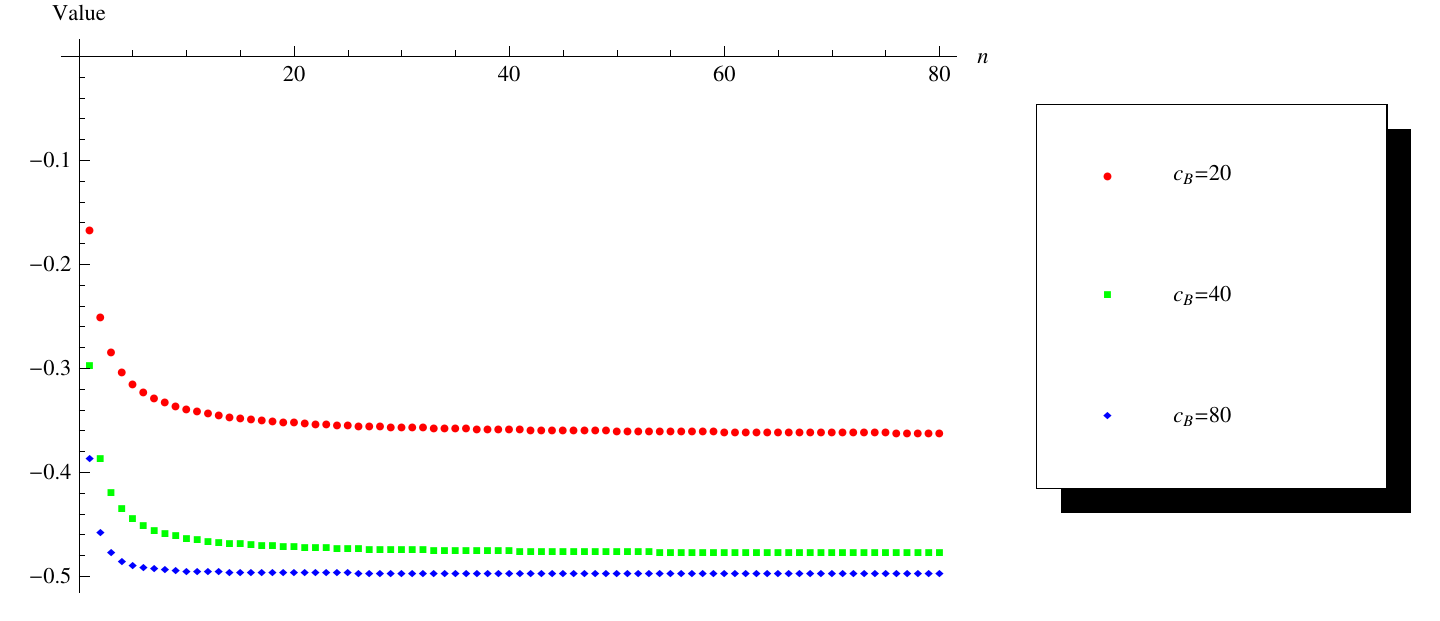}
~\vspace{0cm} \caption{\label{fi:plotvariouscb} Value of $\mathcal{G}$ as a function of $n$  for $c_{A}=10$ and various $c_{B}$.}
\end{figure}

\begin{figure}[h]
\centering
\includegraphics[width=15cm, height=7cm]{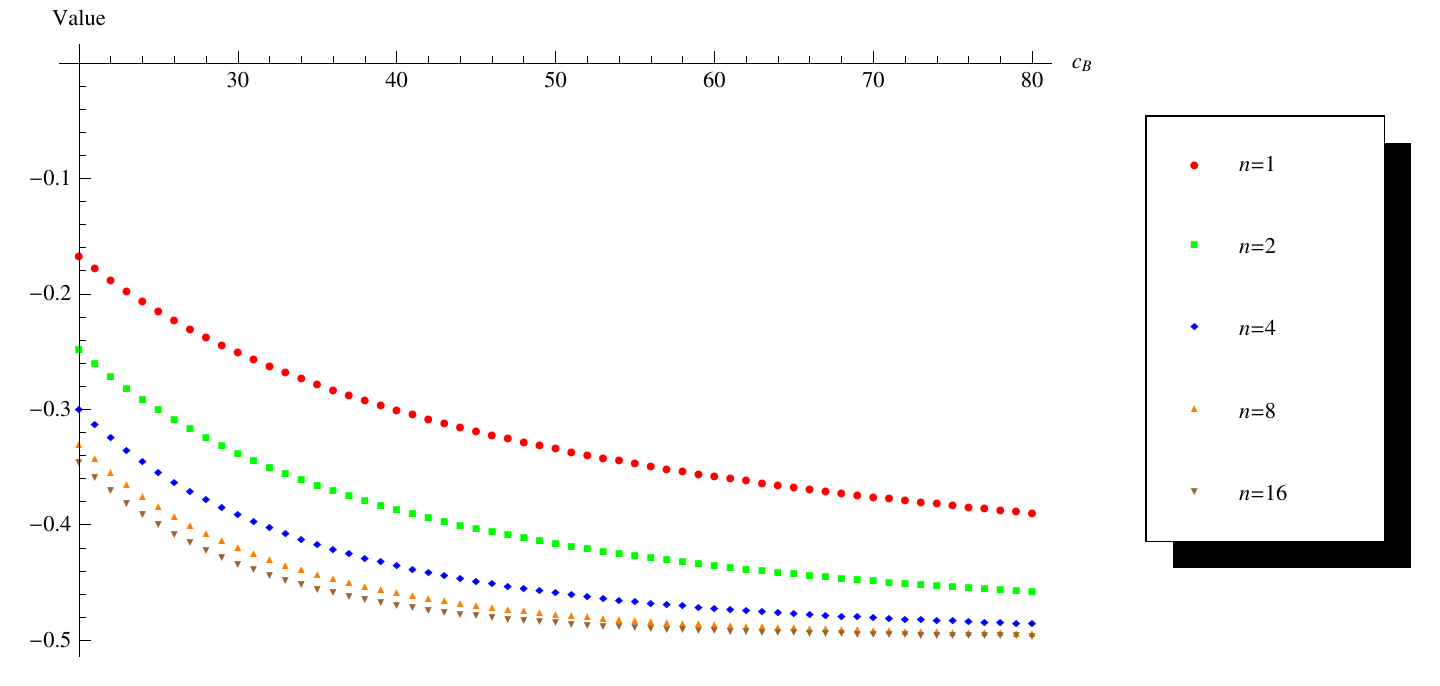}
~\vspace{0cm} \caption{\label{fi:plotvariousn} Value of $\mathcal{G}$ as a function of $c_{B} \ge 20$  for $c_{A}=10$ and various $n$.}
\end{figure}

\end{document}